\documentclass[journal,twoside,web]{ieeecolor}  

\usepackage{lcsys}
\usepackage{graphics} 
\usepackage{epsfig} 
\usepackage{mathptmx} 
\usepackage{times} 
\usepackage{amsmath} 
\usepackage{color}
\usepackage{amssymb}  
\usepackage{mathtools}
\usepackage{verbatim}
\usepackage{enumerate}
\usepackage{algorithm}
\usepackage{algorithmic}
\usepackage{graphicx}
\usepackage{subcaption}
\usepackage{overpic}
\usepackage{cite}
\usepackage{multirow}

\newcounter{definition}
\newenvironment{definition}[1]{\refstepcounter{definition}\par\medskip
\noindent 
\textbf{Definition \thedefinition~(#1)} \em \rmfamily}
{\medskip}
\newenvironment{definition*}{\refstepcounter{definition}\par\medskip
\noindent 
\textbf{Definition \thedefinition.} \em \rmfamily}
{\medskip}

\newcounter{theorem}
{\medskip}

\newcounter{proposition}
\newenvironment{proposition}{\refstepcounter{proposition}\par\medskip
\noindent 
\textbf{Proposition \theproposition.} \em \rmfamily}
{\medskip}

\newenvironment{proposition*}[1]{\refstepcounter{proposition}\par\medskip
\noindent 
\textbf{Proposition \theproposition~(#1)} \em \rmfamily}
{\medskip}

\newcounter{lemma}
\newenvironment{lemma}{\refstepcounter{lemma}\par\medskip
\noindent 
\textbf{Lemma \thelemma.} \em \rmfamily}
{\medskip}

\newenvironment{lemma*}[1]{\refstepcounter{lemma}\par\medskip
\noindent 
\textbf{Lemma \thelemma~(#1)} \em \rmfamily}
{\medskip}

\newcounter{problem}
{\medskip}

\newenvironment{problem*}[1]{\refstepcounter{problem}\par\medskip
\noindent 
\textbf{Problem \theproblem~(#1)} \em \rmfamily}
{\medskip}

\newcounter{remark}
\newenvironment{remark}{\refstepcounter{remark}\par\medskip
\noindent 
\textit{Remark \theremark.} \em \rmfamily}
{\medskip}

\definecolor{darkblue}{rgb}{0.0, 0.0, 0.55}

\newcommand{\nt}{\textcolor{black}{N-type}}
\newcommand{\at}{\textcolor{black}{A-type}}
\newcommand{\nts}{\textcolor{black}{N-types}}
\newcommand{\ats}{\textcolor{black}{A-types}}
\newcommand{\be}{\begin{equation}}
\newcommand{\ee}{\end{equation}}
\newcommand{\nn}{\nonumber}

\newcommand{\vv}{\widehat{v}}

\newcommand{\IA}{I^A}
\newcommand{\IN}{I^N}
\newcommand{\Z}{Z}
\newcommand{\z}{z}

\newcommand{\pz}{\overline{p}}

\newcommand{\piz}{\overline{\pi}}

\newcommand{\extendedversion}[1]{}

\newcommand{\rev}[1]{{\color{black}#1} \color{black}}

\title{\LARGE \bf
Strategizing against Q-learners: A Control-theoretical Approach
}

\author{Yuksel Arslantas, Ege Yuceel, and Muhammed O. Sayin 
\thanks{*This work was supported by The Scientific and Technological Research Council of T\"{u}rkiye (TUBITAK) BIDEB 2232-B International Fellowship for Early Stage Researchers under Grant Number 121C124.}
\thanks{Y. Arslantas, E. Yuceel and M. O. Sayin are with the Department of Electrical and Electronics Engineering,
        Bilkent University, Ankara, T\"{u}rkiye.
        Email: {\tt\small yuksel.arslantas@bilkent.edu.tr, ege.yuceel@ug.bilkent.edu.tr, sayin@ee.bilkent.edu.tr}}}

\begin{document}

\maketitle
\thispagestyle{empty}
\pagestyle{empty}

\begin{abstract}
In this paper, we explore the susceptibility of the independent Q-learning algorithms (a classical and widely used multi-agent reinforcement learning method) to strategic manipulation of sophisticated opponents in normal-form games played repeatedly. We quantify how much strategically sophisticated agents can exploit naive Q-learners if they know the opponents' Q-learning algorithm. To this end, we formulate the strategic actors' interactions as a stochastic game (whose state encompasses Q-function estimates of the Q-learners) as if the Q-learning algorithms are the underlying dynamical system. We also present a quantization-based approximation scheme to tackle the continuum state space and analyze its performance for two competing strategic actors and a single strategic actor both analytically and numerically.
\end{abstract}


\section{INTRODUCTION} 

\IEEEPARstart{T}{he} widespread adoption of (reinforcement) learning algorithms in multi-agent systems has significantly enhanced autonomous systems to tackle complex tasks through learning from interactions within a shared environment. However, strategically sophisticated actors can exploit such algorithms to perform sub-optimally \cite{ref:Huang19,ref:Huang21,ref:Deng19,ref:Vundurthy23,ref:Dong22}. A critical question is \textit{how much} a strategically sophisticated agent can manipulate the opponent's decisions for more payoff in a game if the agent is aware of the opponent's learning dynamic. This strategic act can also have a positive impact on the opponent, depending on how aligned their objectives are.

To address the challenge, we approach the problem from a control-theoretical perspective. We focus on the repeated play of general-sum normal-form games played by two types of agents:  \textit{strategically advanced} (A-type) and \textit{naive} (N-type). Each \nt~\textit{naively} follows the widely-used independent Q-learning (IQL) algorithm as if no other agents exist in the environment \cite{ref:Tan93,ref:Hussain23}. On the other hand, an \at~is a strategically sophisticated agent with complete model knowledge, including underlying game structure, \nts' algorithms, and observations of all actions. We consider \ats~with these advanced capabilities to establish a benchmark for quantifying the performance of other strategic actors with limited capabilities. 

We show that \ats~can leverage their model knowledge to control/drive \nts' algorithms (as if the algorithms are some dynamical system) to maximize the discounted sum of the payoffs collected over an infinite horizon. We can model the interactions among \ats~as stochastic games (SG)--a generalization of Markov decision processes (MDPs) to non-cooperative multi-agent environments \cite{ref:Shapley53}. 
However, the SG has a continuous state space. To address this issue, as an example, we present a quantization-based approximation scheme reducing the problem to a finite SG that can be solved via standard dynamic programming methods when there are two competing \ats~or a single \at. We quantify the approximation level of this quantization scheme and examine its performance numerically.

\textbf{Related Works.} Strategizing against \textit{naive} learning algorithms has been studied for no-regret learning and fictitious play dynamics \cite{ref:Deng19,ref:Dong22,ref:Vundurthy23}. In \cite{ref:Deng19}, the authors show that sophisticated agents can secure the Stackelberg equilibrium value against no-regret learners. 
In \cite{ref:Dong22}, the authors study strategizing against a fictitious player in two-agent games with two actions.
In \cite{ref:Vundurthy23}, the authors prove that sophisticated agents with the knowledge of game matrix can attain better payoffs than the one at Nash equilibrium by solving a linear program for each action of the fictitious player to obtain her own mixed strategy, then play a pure action trajectory to satisfy the desired mixed strategy probabilities. On the other hand, here, we focus on strategizing against the widely-studied IQLs, e.g., see \cite{ref:Wunder10,ref:Tampuu17,ref:Calvano20,ref:Klein21,ref:Hussain23}, and bringing the problem to the control-theoretical framework.

Learning algorithms relying on feedback from the environment (such as Q-learning) can be vulnerable to exploitation through the manipulation of the feedback by some adversaries. 
In a broader sense, our approach relates to the literature on reward poisoning \cite{ref:Huang19,ref:Huang21,ref:Lei23,ref:Zhang20}. The main distinction between our work and this body of literature is the presence of an underlying game structure, which constrains the players' impact on each other. Furthermore, while the general setting in reward poisoning is highly competitive (e.g., zero-sum), we address arbitrary general-sum games. The most relevant studies to our work are those by \cite{ref:Huang19,ref:Huang21}, where the authors investigate the impact of manipulating Q-learner algorithms through the falsification of cost signals and its effect on the algorithm's convergence. Correspondingly, robust variants of the Q-learning against such attacks have also been studied extensively \cite{ref:Sahoo20,ref:Nisioti21,ref:Wang20}.

\textbf{Motivating Example.} Recently, IQLs have been observed to lead to \textit{tacit collusion} that can undermine the competitive nature of the markets \cite{ref:Calvano20,ref:Klein21,ref:Hansen21,ref:Banchio22}. For example, in \cite{ref:Banchio22}, the authors study the collusive behavior of IQLs in the widely-studied prisoner's dilemma game (where agents have two actions: `cooperate' and `deflect'). They observe that Q-learners can learn to collude in cooperation even though `cooperate' is always an irrational choice as `deflect' dominates the `cooperate' strategy. New regulations are needed to prevent such collusive behavior \cite{ref:Calvano20}. Here, we numerically examine the susceptibility of IQLs to strategic manipulation in prisoner's dilemma. Such vulnerabilities can incentivize businesses not to use IQLs with the hope of tacit collusion, as their naive approach might be exploited by the other business with strategically sophisticated algorithms.

\textbf{Contributions.} Our main contributions include: 
\begin{itemize}
    \item To the best of our knowledge, our work is \textit{the first} to model the vulnerability of IQLs to strategic manipulations by advanced players as an SG as if the IQLs are the underlying dynamic environment.
    \item  The continuum state space in the SGs poses a technical challenge. However, we demonstrate the Lipschitz continuity of the value functions for the SGs based on the specifics of the IQL updates. This characterization shows that the SG formulated is relatively \textit{well-behaved} as the approximation methods, such as the value-based ones, can perform effectively.
\end{itemize}

\textbf{Organization.} The paper is organized as follows: We formulate the strategic actors' interactions as an SG in Section \ref{sec:preliminary}. We present our solution in Section \ref{sec:approximation}. We provide illustrative examples and conclude the paper, resp., in Sections \ref{sec:example} and \ref{sec:discussion}. Three appendices include the proofs of technical lemmas.

\begin{figure*}[t!]
    \centering
    \includegraphics[width=2\columnwidth]{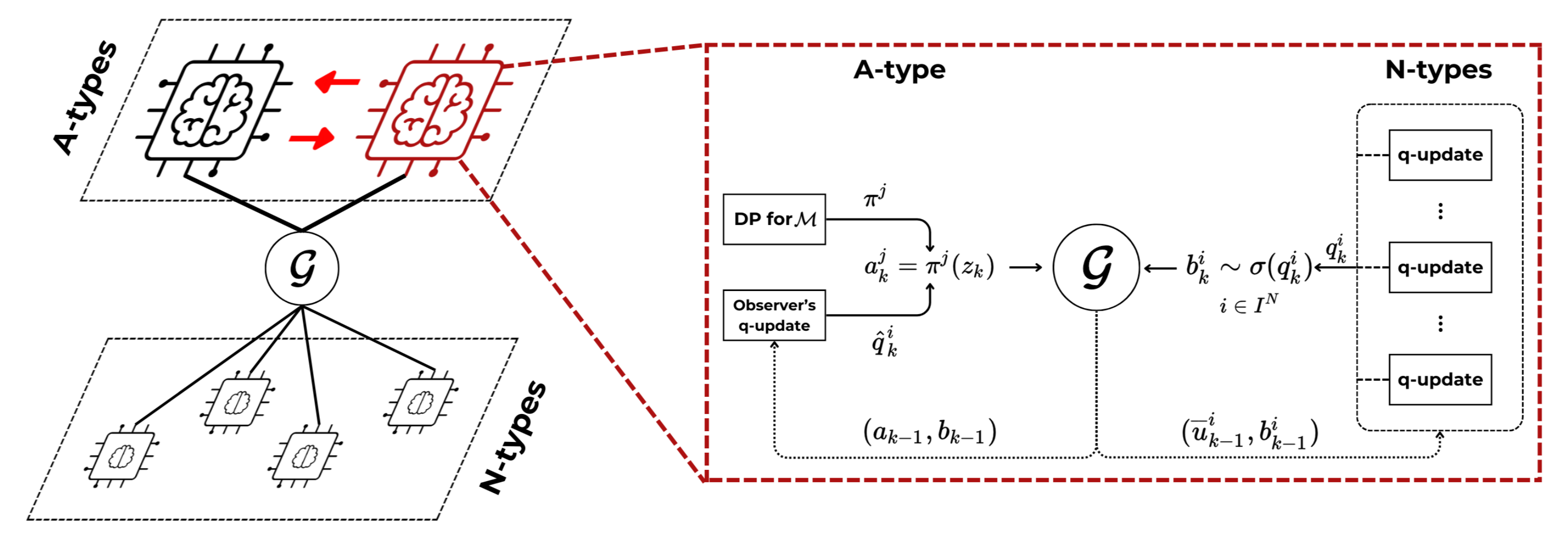}
    \caption{Illustration of the interaction between \nts~and \ats~across the repeated play of the normal-form game $\mathcal{G}$. \ats~compete against each other while considering \nts' q-function estimates as an underlying state of an SG. \ats~can use dynamic programming to solve the SG and use trackers to track \nts' q-function estimates by leveraging the complete model knowledge.}\label{fig:model}
\end{figure*}

\section{STRATEGIC ACTORS AGAINST Q-LEARNERS}\label{sec:preliminary}
Consider a normal-form game $\mathcal{G}$ characterized by the tuple $\langle \{A^j\}_{j\in\IA}, \{B^i\}_{i\in\IN}, \{u^j\}_{j\in\IA}, \{\overline{u}^i\}_{i\in\IN} \rangle$, where $\IA$ and $\IN$ denote, resp., the index set of \at~and \nt~agents.
\begin{itemize}
    \item \textbf{\nts} are \textit{naive} players following the IQL algorithm based on the local information (i.e., payoff received) as if there are no other agents. 
    \item \textbf{\ats} are \textit{advanced} players with complete knowledge of the underlying game and the employed IQLs. They can also access the joint actions of all players.
\end{itemize}
Let $A^j$ and $B^i$ denote the \at~$j$'s and \nt~$i$'s finite action sets. The joint action spaces of \ats~and \nts~are denoted by $A \coloneqq \bigtimes_{j\in\IA} A^j$ and $B\coloneqq \bigtimes_{i\in\IN} B^i$, respectively. The payoff functions are $u^j: A\times B \rightarrow \mathbb{R}$ and $\overline{u}^i: A \times B \rightarrow \mathbb{R}$. \ats~and \nts~are repeatedly playing $\mathcal{G}$ by taking actions simultaneously over stages $k = 0,1,\ldots$ as follows:
    
\textbf{\nt:} Let $q_k^i:B^i\rightarrow\mathbb{R}$ denote \nt~$i$'s q-function estimate at stage $k$. Then, the agent takes action according to the softmax function:
\begin{flalign}\label{eq:softmax}
\sigma(q_k^i)(b^i) = \frac{\exp(q_k^i(b^i)/\tau)}{\displaystyle\sum_{\tilde{b}^i \in B^i} \exp(q_k^i(\tilde{b}^i)/\tau)}\quad\forall b^i\in B^i
\end{flalign}
for some $\tau>0$ controlling the level of exploration and $\sigma(q_k^i)(b^i)\in (0,1]$ denotes the probability that they take action $b^i \in B^i$. Let $a_k=(a_k^j)_{j\in\IA}$ and $b_k=(b_k^i)_{i\in\IN}$ denote the joint actions of \ats~and \nts~at stage $k$, respectively. Given the payoff $\overline{u}_k^i = \overline{u}^i(a_k,b_k)$, \nt~$i$ updates the q-function according to
\begin{flalign}\label{eq:qupdate}
        q_{k+1}^i(b^i) = 
        \begin{cases}
                q_k^i(b^i) + \alpha (\overline{u}_k^i - q_k^i(b^i)) &\; \text{if } b^i=b_k^i \\
                q_k^i(b^i) &\; \text{o.w}
        \end{cases}
\end{flalign}
with some step size $\alpha\in (0,1)$ and the initial vector $q_0^i(b^i) = 0$ for all $b^i$.\footnote{We interchangeably view q-function as a function $q^i:B^i\rightarrow\mathbb{R}$ and a vector $q^i\in \mathbb{R}^{|B^i|}$ since $B^i$ is a finite set.} In other words, \nts~act myopically as if the other players play according to some stationary strategy across the repeated play of the underlying game $\mathcal{G}$.

\textbf{\at:} The advanced players can track \nts' q-function estimates by knowing the deployed IQLs, as described in \eqref{eq:qupdate}, and observing the actions of all agents. Therefore, \ats~do not myopically act as they are aware that the other players do not necessarily play according to some stationary strategy. Instead, they strategically aim to maximize the discounted sum of the payoffs they can collect across the repeated play of the underlying game $\mathcal{G}$, i.e.,
        \begin{flalign}\label{eq:aliceobjective1}
                \max_{\{a_k^i\}_{k=1}^\infty} \mathrm{E} \left[ \sum_{k=0}^{\infty} \gamma^k u^i\left(a_k,b_k\right)\right]
                \end{flalign}
        for some discount factor $\gamma\in (0,1)$.

\at~$i$ takes actions to achieve \eqref{eq:aliceobjective1} based on the history $\{a_0,b_0,\ldots,a_{k-1},b_{k-1}\}$ of the game and the knowledge that $q_k^i$'s evolve according to \eqref{eq:qupdate}. However, this is not a well-defined optimization problem that can be solved introspectively due to its dependence on the other players' play. 

Although the information set $\{a_0,b_0,\ldots,a_{k-1},b_{k-1}\}$ grows unboundedly in time, \nts' actions depend only on their q-function estimates. In other words, conditioned on the current q-function estimates and the current action profile of \ats, the next q-function estimates, and therefore, \nts' play is independent of the game history.
\ats~can leverage this Markov property in their objective \eqref{eq:aliceobjective1} by viewing the q-function estimates as some underlying controlled Markov process. Therefore, we can model \ats' interactions with each other as the following stochastic game whose state is \nts' q-function estimates.

\begin{problem*}{Stochastic Game Formulation}\label{prob:mdp}
    Consider an $|\IA|$-agent SG characterized by the tuple $\mathcal{Z}=\langle \Z,(A^j,r^j)_{j\in\IA},\gamma,p\rangle$. The state space $\Z\coloneqq \prod_{i\in\IN}Q^i$ is a compact set, where $Q^i\subset\mathbb{R}^{|B^i|}$ corresponds to all possible q-function estimates of \nt~$i\in \IN$. The reward function $r^j:\Z\times A\rightarrow \mathbb{R}$ is given by
\be\label{eq:rz}
r^j(\z,a) = \mathrm{E}_{b\sim \overline{\pi}(z)} [u^j(a,b)]\quad\forall (\z,a)\in \Z\times A,
\ee
where 
\be\label{eq:pizz}
\overline{\pi}(\{q^i\}_{i\in\IN})(\{b^i\}_{i\in\IN}) \coloneqq \prod_{i\in\IN}\sigma(q^i)(b^i)
\ee
for all $q^i\in Q^i$, $b^i\in B^i$ and $i\in\IN$.\footnote{By concatenation of $\{q^i\}_{i \in \IN}$, we can construct state $z$. Since $b$ is the action profile, $\overline{\pi}(\{q^i\}_{i\in\IN})(\{b^i\}_{i\in\IN})$ is equivalent to $\piz(z)(b)$.}
The transition kernel $p(\cdot\mid\cdot)$ determines the evolution of transformed states according to the update rule \eqref{eq:qupdate}. 
\end{problem*}

Recall that \ats~know \nts' IQLs \eqref{eq:qupdate} and the underlying game model $\mathcal{G}$ and can observe the actions $(a_k,b_k)$ of all agents. Therefore, they can track the q-function estimates perfectly, although they do not have direct access to them. For example, \ats~can deploy a q-tracker $\widehat{q}_k^i$ specific to each \nt~$i$ and evolving according to
\begin{subequations}\label{eq:qtracker}
\begin{flalign}
    &\widehat{q}_{k+1}^i(b^i_k) = \widehat{q}_k^i(b^i_k) + \alpha (\overline{u}^i (a_k,b_k) - \widehat{q}_k^i(b^i_k)) \\
    &\widehat{q}_{k+1}^i(b^i) = \widehat{q}_k^i(b^i)\quad\forall b^i\neq b_k^i
\end{flalign}
\end{subequations}
with $\widehat{q}_0^i = q_0^i$. Therefore, we have $\widehat{q}_k^i\equiv q_k^i$ for all $k$ and \ats~can construct the state $z_k=(\widehat{q}_k^i)_{i\in\IN}$ of the underlying SG $\mathcal{Z}$.

Let each \at~$j$ follow Markov stationary strategy $\pi^j:Z\rightarrow\Delta(A^j)$ in the SG $\mathcal{Z}$.\footnote{The probability simplex over a finite set $A$ is denoted by $\Delta(A)$.} Given the strategy profile $\pi=(\pi^j)_{j\in\IA}$, \at~$j$'s utility \eqref{eq:aliceobjective1} can be rewritten as
\be
U^j(\pi^j,\pi^{-j}) := \mathrm{E}_{a_k\sim\pi(z_k)}\left[\sum_{k=0}^{\infty} \gamma^k r^j(z_k,a_k)\right]
\ee
for the SG $\mathcal{Z}$. As we pointed out earlier, finding the best strategy in games is generally not a well-defined optimization problem. Therefore, we focus on the important special case of two competing \ats. Let $\IA = \{j_1,j_2\}$ and \ats~have opposite payoffs $u^{j_1}+u^{j_2}\equiv 0$ such that $U^{j_1}+U^{j_2}\equiv 0$.

\begin{algorithm}[t]
    \caption{\at~$j$}
    \label{tab:A}
    \begin{algorithmic}
        \STATE \textbf{input:} $\pi^j$ for the SG $\mathcal{Z}$ and the initials $\{\widehat{q}_{0}^i\}_{i\in\IN}$
        \FOR{each stage $k=0,1,\ldots$}  
        \STATE construct $z_k = \{\widehat{q}_k^i\}_{i\in\IN}$
        \STATE take action $a_k^j\sim \pi^j(z_k)$ simultaneously with the others
        \STATE receive payoff $u^j(a_k,b_k)$ (corresponding to $r^j(z_k,a_k)$ in expectation)
        \STATE observe $(a_k,b_k)$
        \STATE update the q-trackers according to \eqref{eq:qtracker}
        \ENDFOR
    \end{algorithmic}
\end{algorithm}

\begin{algorithm}[t]
    \caption{\nt~$i$}
    \label{tab:N}
    \begin{algorithmic}
    \STATE \textbf{input:} the initial $q_{0}^i$
    \FOR{each stage $k=0,1,\ldots$}  
    \STATE take action $b_k^i\sim \sigma(q_k^i)$ simultaneously with the others
    \STATE receive payoff $\overline{u}_k^i =\overline{u}^i\left(a_k,b_k\right)$
    \STATE update the q-function estimate according to \eqref{eq:qupdate}
    \ENDFOR
    \end{algorithmic}
\end{algorithm}

\begin{definition}{Markov Stationary Equilibrium} We say that a Markov stationary profile $\pi=(\pi^j:Z\rightarrow\Delta(A^j))_{j\in\IA}$ is Markov stationary equilibrium provided that
    \be 
    U^j(\pi^j,\pi^{-j}) \geq U^j(\tilde{\pi}^j,\pi^{-j}) \quad\forall \tilde{\pi}^j,
    \ee
    where $\pi^{-j}=\{\pi^\ell\}_{\ell\neq j}$ \cite{ref:Shapley53}.
\end{definition}

\begin{proposition}
        There always exists Markov stationary equilibrium for two-agent zero-sum SGs $\mathcal{Z}=\langle \Z,(A^j,r^j)_{j\in\IA},\gamma,\pz\rangle$.
\end{proposition}

\begin{proof}
The proof follows from \cite[Theorem 6.2.12]{ref:Puterman14} based on the observation that the minimax function is non-expansive like the maximum function and the set of all possible q-function estimates, $Z$, is a Polish space as a compact subset of $\mathbb{R}^{|B|}$ and the set of actions, $A$, is finite and state-invariant. 
\end{proof}

\begin{remark}\label{rem:single}
If there is only a single \at, then the SG $\mathcal{Z}$ reduces to an MDP, for which there also exists a Markov stationary solution.
\end{remark}

Given the Markov stationary strategy $\pi^j$ for the SG $\mathcal{Z}$, \at~$j$ interacts with other agents as illustrated in Fig. \ref{fig:model}. Detailed descriptions of the \ats' and \nts' interactions across the repeated play of the underlying game $\mathcal{G}$ are tabulated, resp., in Algorithms \ref{tab:A} and \ref{tab:N}.

\section{SOLVING THE STOCHASTIC GAME}\label{sec:approximation}

A critical challenge for finding the best strategy in the two-agent zero-sum SG $\mathcal{Z}$ is the continuum state space $Z$. Based on the specifics of the IQLs, we can show that the SG $\mathcal{Z}$ is relatively well-behaved for effective approximations. Particularly, let $v_{\kappa}^j:{Z}\rightarrow \mathbb{R}$ denote the values of the states in {$Z$} at the initial stage of a $(\kappa+1)$-length horizon for the best decision rules in $\mathcal{Z}$. More explicitly, for each {$z\in Z$}, we have
\begin{subequations}\label{eq:v}
\begin{flalign}
&v_{\kappa}^{j_1}({z}) \coloneqq \max_{\mu^{j_1}}\min_{\mu^{j_2}} \;\mathrm{E}_{(a^{j_1},a^{j_2})\sim (\mu^{j_1},\mu^{j_2})}[Q_{\kappa}^{j_1}(z,a^{j_1},a^{j_2})],\\
&v_{\kappa}^{j_2}({z}) \coloneqq \max_{\mu^{j_2}}\min_{\mu^{j_1}} \;\mathrm{E}_{(a^{j_1},a^{j_2})\sim (\mu^{j_1},\mu^{j_2})}[Q_{\kappa}^{j_2}(z,a^{j_1},a^{j_2})],
\end{flalign}
\end{subequations}
where
\be
Q^{j}_{\kappa}(z,a) \coloneqq r^j({z},a) + \gamma \int_{{Z}} v_{\kappa-1}^j({\tilde{z}})p(d{\tilde{z}}\mid {z},a)
\ee
for all $\kappa>0$, $j\in\IA$ and $(z,a)\in Z\times A$. The initial $Q_0^j({z},a) = r^j({z},a)$ for all $(z,a)$.

\begin{proposition}\label{pro:value_lip}
    The minimax value function $v_\kappa^j(z)$ for each $j\in\IA$ is bounded by 
    \be
    M_{\kappa}:= \frac{1-\gamma^{\kappa+1}}{1-\gamma}\max_{(\tilde{j},a,b)}|u^{\tilde{j}}(a,b)|
    \ee
    for all $z\in Z$ and $\kappa \geq 0$. The value function is also $L_\kappa$-Lipschitz with respect to $\|\cdot\|_1$, where 
    \be
    L_{\kappa} = \gamma L_{\kappa-1} + L_0\frac{1-\gamma^{\kappa+1}}{1-\gamma}
    \ee
    for all $\kappa> 0$ and $L_0 \coloneqq \frac{1}{\tau} \sqrt{|B|} \max_{(\tilde{j},a,b)} |u^{\tilde{j}}(a,b)|$.
\end{proposition}

{
\begin{proof}
    The boundedness follows from the definitions \eqref{eq:rz} and \eqref{eq:v} since we can bound $v_{\kappa}^j(\cdot)$ by
    \begin{align}
        |v_{\kappa}^j({z})|&\leq \sum_{k=0}^{\kappa} \gamma^k \max_{({z},a)}|r^j({z},a)|\nn\\
        &\leq \sum_{k=0}^{\kappa} \gamma^k \max_{(\tilde{j},a,b)}|u^{\tilde{j}}(a,b)|\nn\\
        &\leq \frac{1-\gamma^{\kappa+1}}{1-\gamma}\max_{(\tilde{j},a,b)}|u^{\tilde{j}}(a,b)|=M_{\kappa}.\label{eq:vbound}
        \end{align}
        
Next, we first present the following lemma playing an important role in showing the Lipschitz continuity of the value function.

\begin{lemma}\label{lem:pi}
The joint strategy $\overline{\pi}(\cdot)$ of \nts, as described in \eqref{eq:pizz}, satisfies
\be
\|\overline{\pi}(z) - \overline{\pi}(\overline{z})\|_1 \leq \frac{\sqrt{|B|}}{\tau}\cdot \|z-\overline{z}\|_1
\ee
for all $z,\overline{z} \in Z$.
\end{lemma}        
        
    The second part of the proof follows from induction. Firstly, for $\kappa=0$, the non-expansiveness of the minimax function leads to
    \begin{align}\label{eq:q_lip1}
        |v_{0}^j({z})-&v_{0}^j({\overline{{z}}})| \nn\\
        &\leq  \max_{\mu^{j}} \min_{\mu^{-j}} \;\mathrm{E}[|r^j({z},a^j,a^{-j})-r^j({\overline{z}},a^j,a^{-j})|],
    \end{align}
    where the expectation is taken with respect to $a^j\sim \mu^j$ and $a^{-j}\sim \mu^{-j}$.
    Then, we can use the following lemma showing the Lipschitz continuity of the reward functions $r^j(\cdot)$'s.
    
    \begin{lemma}\label{lemma:reward_lip}
        The reward function $r^j({z},a)$ is Lipschitz continuous in {$z\in Z$}. For each $a\in A$, we have
        \begin{flalign}\label{eq:Lipschitz_reward}
            |r^j({z},a)-r^j({\overline{z}},a)| \leq L_0 \cdot \| \rev{z - \overline{z}} \|_1
        \end{flalign}
        for all ${z,\overline{z} \in Z}$, where $L_0$ is as described in Proposition \ref{pro:value_lip}.
    \end{lemma}

    By \eqref{eq:q_lip1}, Lemma \ref{lemma:reward_lip} yields that
    \begin{flalign}
        |v_0^j(z) - v_0^j(\overline{z})| &\leq \max_{\mu^{j}} \min_{\mu^{-j}} \mathrm{E}[|r^j({z},a)-r^j({\overline{z}},a)|] \nn\\&\leq L_0 \cdot \|z-\overline{z}\|_1.
    \end{flalign}
    Next, we focus on $\kappa>0$ and assume that $v_{\kappa-1}^j(\cdot)$ is $L_{\kappa-1}$-Lipschitz with respect to $\|\cdot\|_1$ to show that $v_\kappa^j(\cdot)$ is $L_{\kappa}$-Lipschitz. The non-expansiveness of the minimax function and the triangle inequality yield that
    \begin{flalign}
        &|v_{\kappa}^j(z)-v_{\kappa}^j(\overline{z})| \leq  \max_{\mu^{j}} \min_{\mu^{-j}} \mathrm{E} [|r^j(z,a)-r^j(\overline{z},a)|]\nn \\
        &+ \gamma \max_{\mu^{j}} \min_{\mu^{-j}} \mathrm{E} \left[\left| \int_Z v_{\kappa-1}^j(\tilde{z})p(d\tilde{z}|z,a) - \int_Z v_{\kappa-1}^j(\tilde{z})p(d\tilde{z}|\overline{z},a) \right|\right].\nn
    \end{flalign}
    For the first term on the right-hand side, we can invoke Lemma \ref{lemma:reward_lip}. On the other hand, for the second term, we can use the following lemma.
    
    \begin{lemma}\label{lemma:prob_bound}
        Let $\eta: Z \rightarrow \mathbb{R}$ be bounded $|\eta(z)|\leq M$ for all $z\in  Z$ and $K$-Lipschitz with respect to $\|\cdot\|_1$. Then we have
        \begin{flalign}\label{eq:prob_lip}
            \left| \int_{Z} \eta(\tilde{z}) p(d\tilde{z}|z,a) - \int_{Z} \eta(\tilde{z}) p(d\tilde{z}|\overline{z},a) \right| \leq \overline{K} \cdot \|z-\overline{z}\|_1,
        \end{flalign}
        where $\overline{K} = K + \frac{\sqrt{|B|}}{\tau}\cdot M$.
    \end{lemma}

    Based on the boundedness of the value function \eqref{eq:vbound} and the assumption that $v_{\kappa-1}^j$ is $L_{\kappa-1}$-Lipschitz, we obtain
    \begin{align}
    |v_{\kappa}^j(z)-v_{\kappa}^j(\overline{z})|&\leq 
    \left(L_0+\gamma\left(L_{\kappa-1} + \frac{\sqrt{|B|}}{\tau}M_{\kappa-1}\right)\right)\cdot\|z-\overline{z}\|_1\nn\\
        &= L_{\kappa}\cdot \|z-\overline{z}\|_1
    \end{align}
    for all $z,\overline{z}\in Z$, where the equality follows from the definitions of $M_{\kappa-1}$, $L_{\kappa}$, and $L_0$ in the proposition statement. This completes the proof.
\end{proof}}

The Lipschitz continuity of the minimax value function is promising to effectively approximate the underlying SG via the value-based approximation methods. Various methods exist for approximating MDPs with large state-action spaces. For instance, state aggregation \cite{ref:Ren02,ref:Ortner07} involves combining similar states with respect to cost and transition probabilities into meta-states. On the other hand, in approximate linear programming \cite{ref:DeFarias03,ref:DeFarias04}, preselected basis functions are used to fit the value function of the MDP. Approximate value iteration and approximate policy iteration \cite{ref:Tsitsiklis96,ref:Bertsekas96,ref:Farahmand10} also use preselected basis functions to determine the value and policy, respectively. Neuro-dynamic programming \cite{ref:Bertsekas96} leveraging the power of approximation structures, such as neural networks, is also an important method to approximate MDPs. 

In this paper, as an \textit{example} proof of concept, we focus on quantization-based approximations (e.g., see \cite{ref:Saldi18} for a detailed review) due to their interpretability by transforming $\mathcal{Z}$ to a finite SG, as in the original definition in \cite{ref:Shapley53}.  To this end, we consider a quantization mapping $\Phi:Z\rightarrow D$ for some finite set $D\subset Z$ such that $Z_d := \{z\in Z: \Phi(z) = d\}$ satisfies $Z_d \cap Z_{d'} = \varnothing$ for all $d\neq d'$ and $\bigcup_{d\in D}Z_d = Z$. In the following, we quantify the approximation error induced by such quantization schemes based on Proposition \ref{pro:value_lip}.

\begin{proposition}
    Given the quantization mapping $\Phi$, assume that there exists $\Delta>0$ such that $\|z-\Phi(z)\|_1\leq \Delta$ for all $z\in Z$. Then, for each $j\in\IA$, we have
     \begin{flalign}\label{eq:result}
         \|v_{\kappa}^j-\vv_{\kappa}^j\|_{\infty} \leq \frac{\Delta\sqrt{|B|}}{\tau(1-\gamma)^3}\max_{(\tilde{j},a,b)}|u^{\tilde{j}}(a,b)|\quad\forall \kappa\geq 0,
     \end{flalign}
     where $v_{\kappa}^j$ is described in \eqref{eq:v}, and 
     \be\label{eq:vhat}
     \widehat{v}_{\kappa}^j(z) \coloneqq \max_{\mu^j}\min_{\mu^{-j}} \;\mathrm{E}_{(a^j,a^{-j})\sim(\mu^j,\mu^{-j})}[\widehat{Q}_k^j(z,a^j,a^{-j})],
     \ee
     where
    \be
        \widehat{Q}_{\kappa}^j(z,a) \coloneqq r^j(\Phi(z),a) + \gamma \int_{Z} \vv_{\kappa-1}^j(\tilde{z})p(d\tilde{z}\mid \Phi(z),a).
    \ee
\end{proposition}

{
\begin{proof}
    We first focus on the case $\kappa=0$. By the definitions of $v_0^j$ and $\vv_0^j$, we have
    \begin{align}
    |v_0^j(z) - \vv_0^j&(z)|= \big|\max_{\mu^{j}} \min_{\mu^{-j}} \mathrm{E}[r^j(z,a^j,a^{-j})]\nn\\
    &\hspace{.5in}-\max_{\mu^{j}} \min_{\mu^{-j}} \mathrm{E}[r^j(\Phi(z),a^j,a^{-j})]\big|\nn\\
    &\leq \max_{\mu^{j}} \min_{\mu^{-j}} \mathrm{E}\left[\left|r^j(z,a^j,a^{-j}) - r^j(\Phi(z),a^j,a^{-j})\right|\right]\nn \\
    &\leq L_0 \cdot \|z-\Phi(z)\|_1\quad\forall z\in Z,\label{eq:vdiff0}
    \end{align}
    where we use the non-expansiveness of the minimax function and Lemma \ref{lemma:reward_lip}.
    Next, we focus on $\kappa>0$. Based on the definitions of $v_{\kappa}^j$ and $\vv_{\kappa}^j$, resp., as described in \eqref{eq:v} and \eqref{eq:vhat}, the non-expansiveness of the minimax function and the triangle inequality yield that
    \begin{flalign}
    |v_{\kappa}^j(z)-\vv_{\kappa}^j(z)| &\leq \max_{\mu^{j}} \min_{\mu^{-j}} \mathrm{E}\left[|r^j(z,a^j,a^{-j})-r^j(\Phi(z),a^j,a^{-j})|\right] \nn\\
    &+ \gamma \max_{\mu^{j}} \min_{\mu^{-j}} \mathrm{E} \Big[\Big|\int_{Z} v_{\kappa-1}^j(\tilde{z})p(d\tilde{z}\mid z,a^j,a^{-j})\nn\\
    &\hspace{.1in} -\int_{Z} \vv_{\kappa-1}^j(\tilde{z})p(d\tilde{z}\mid \Phi(z),a^j,a^{-j})\Big|\Big].\nn
    \end{flalign}
    
    We can bound the first term on the right-hand side by using Lemma \ref{lemma:reward_lip}. Therefore, we focus on the second term. 
    However, in the second term, not only the integrands $v_{\kappa-1}^j$ and $\vv_{\kappa-1}^j$ but also the measures $p(\cdot\mid z,a)$ and $p(\cdot\mid \Phi(z),a)$ are different. To address this issue, we can add and subtract the term $\int_{Z} v_{\kappa-1}^j(\tilde{z})p(d\tilde{z}\mid \Phi(z),a)$. Then, by the triangle inequality, we have
    \begin{align}
    &\left|\int_{Z} v_{\kappa-1}^j(\tilde{z})p(d\tilde{z}\mid z,a)-\int_{Z} \vv_{\kappa-1}^j(\tilde{z})p(d\tilde{z}\mid \Phi(z),a)\right|\nn\\
    &\hspace{.2in}\leq \left|\int_{Z} v_{\kappa-1}^j(\tilde{z})p(d\tilde{z}\mid z,a)-\int_{Z} v_{\kappa-1}^j(\tilde{z})p(d\tilde{z}\mid \Phi(z),a)\right|\nn\\
    &\hspace{.4in}+\left|\int_{Z} (v_{\kappa-1}^j(\tilde{z})-\vv_{\kappa-1}^j(\tilde{z}))p(d\tilde{z}\mid \Phi(z),a)\right|.\label{eq:intvdiff}
    \end{align}
    To bound the first term on the right-hand side of \eqref{eq:intvdiff}, we can invoke Proposition \ref{pro:value_lip} and Lemma \ref{lemma:prob_bound} such that
    \begin{align}
        &\left|\int_{Z} v_{\kappa-1}^j(\tilde{z})p(d\tilde{z}\mid z,a)-\int_{Z} v_{\kappa-1}^j(\tilde{z})p(d\tilde{z}\mid \Phi(z),a)\right|\nn\\
        &\hspace{.8in}\leq \left(L_{\kappa-1} + \frac{\sqrt{|B|}}{\tau} M_{\kappa-1}\right)\|z-\Phi(z)\|_1\nn\\
        &\hspace{.8in}= \left(L_{\kappa-1} + L_0\frac{1-\gamma^{\kappa}}{1-\gamma}\right)\|z-\Phi(z)\|_1,\label{eq:bound1}
        \end{align}
        where the equality follows from the definitions of $L_0$ and $M_{\kappa-1}$. On the other hand, we can bound the second term on the right-hand side of \eqref{eq:intvdiff} by
        \begin{align}
        &\left|\int_{Z} (v_{\kappa-1}^j(\tilde{z})-\vv_{\kappa-1}^j(\tilde{z}))p(d\tilde{z}\mid \Phi(z),a)\right| \nn\\
        &\hspace{1.5in}\leq \|v_{\kappa-1}^j-\vv_{\kappa-1}^j\|_{\infty}.\label{eq:bound2}
        \end{align}
        Then, by Lemma \ref{lemma:reward_lip} and the bounds \eqref{eq:bound1}, \eqref{eq:bound2}, we obtain 
        \begin{flalign}
            |{v}_{\kappa}^j(z)-\vv_{\kappa}^j&(z)| \leq \left(L_0 + \gamma L_{\kappa-1} + \gamma L_0 \frac{1-\gamma^{\kappa}}{1-\gamma}\right)\|z-\Phi(z)\|_1\nn \\
            &\hspace{.4in}+ \gamma \|v_{\kappa-1}^j-\vv_{\kappa-1}^j\|_{\infty}\nn
            \\
            &=L_{\kappa} \cdot \|z-\Phi (z)\|_1 + \gamma \|v_{\kappa-1}^j - \vv_{\kappa-1}^j\|_\infty.
        \end{flalign}
        Given $\Delta \geq \|z-\Phi(z)\|_1$ for all $z$, we have
        \begin{subequations}\label{eq:v_inf_bound} 
        \begin{flalign}
            &\|v_0^j-\vv_0^j\|_{\infty}\leq L_0 \cdot \Delta\\
            &\|{v}_{\kappa}^j-\vv_{\kappa}^j\|_\infty \leq L_{\kappa} \cdot \Delta + \gamma \|v_{\kappa-1}^j - \vv_{\kappa-1}^j\|_\infty\quad\forall \kappa>0
        \end{flalign}
        \end{subequations}
        The recursion \eqref{eq:v_inf_bound} on $\|v_{\kappa}^j - \vv_{\kappa}^j\|_\infty$ implies that
        \be\label{eq:vvbound}
        \|v_{\kappa}^j - \vv_{\kappa}^j\|_\infty \leq \Delta \cdot \sum_{k=0}^{\kappa} \gamma^{\kappa-k} L_{k}.
        \ee
        On the other hand, the recursion on $\{L_{\kappa}\}_{\kappa\geq 0}$ implies that 
        \begin{align}\label{eq:Lbound}
        L_{\kappa} = L_0\sum_{\ell = 0}^{\kappa}(\ell+1)\gamma^{\ell}.
        \end{align}
        After some algebra, \eqref{eq:vvbound} and \eqref{eq:Lbound} yield \eqref{eq:result}.
\end{proof}}




\begin{figure}[t!]
    \centering
    \includegraphics[width=.9\columnwidth]{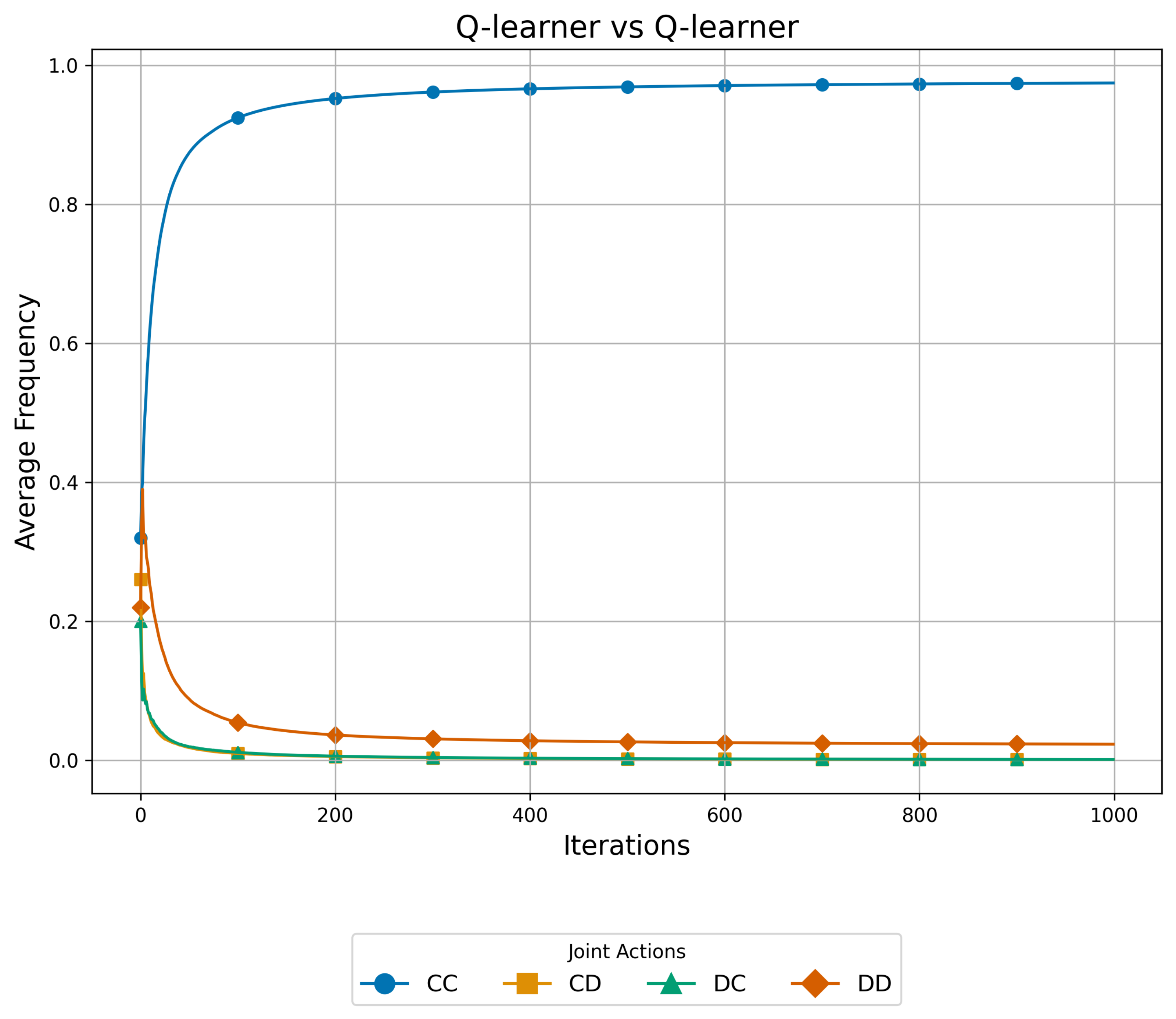}
    \caption{The evolution of the empirical averages of the action profiles for Q-learner vs Q-learner in the prisoner's dilemma.}
    \label{fig:pd1}
\end{figure}

\begin{figure}[t!]
    \centering
    \includegraphics[width=.9\columnwidth]{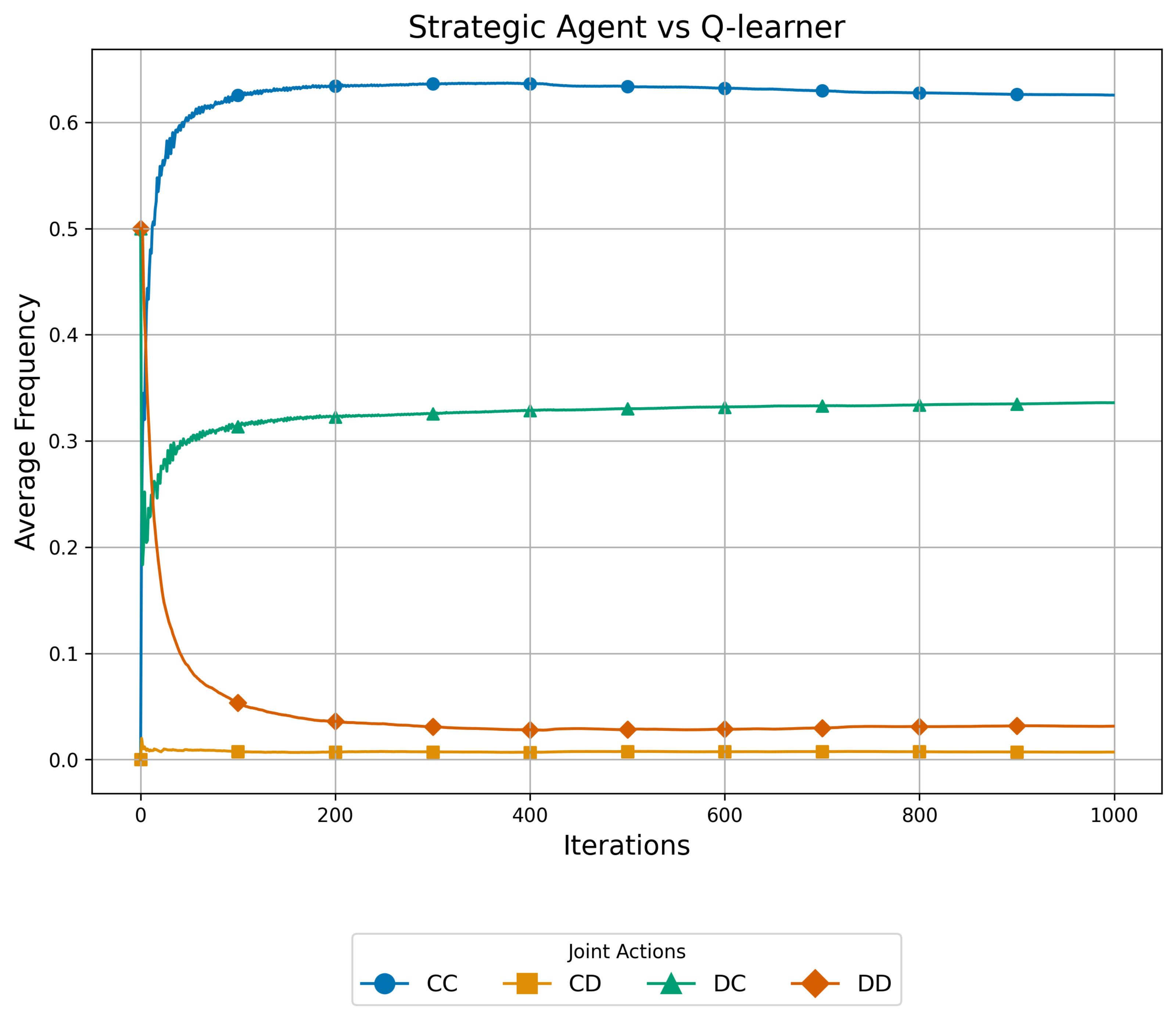}
    \caption{The evolution of the empirical averages of the action profiles for strategic actor vs Q-learner in the prisoner's dilemma.}
    \label{fig:pd2}
\end{figure}

\begin{figure}[t!]
    \centering
    \includegraphics[width=.9\columnwidth]{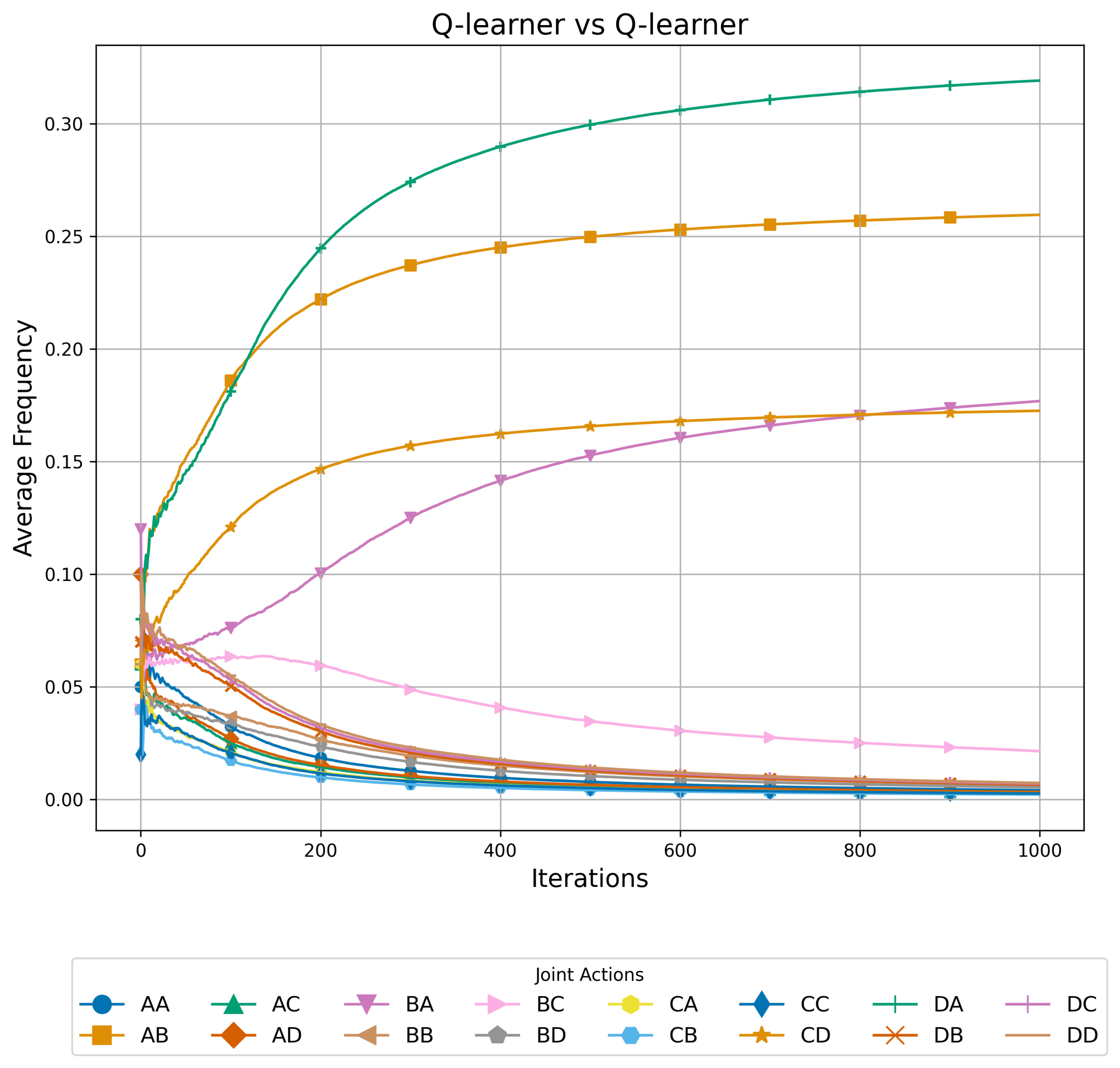}
    \caption{The evolution of the empirical averages of the action profiles for Q-learner vs Q-learner in the two-agent four-action zero-sum game.}
    \label{fig:zero-sum1}
\end{figure}

\begin{figure}[t!]
    \centering
    \includegraphics[width=.9\columnwidth]{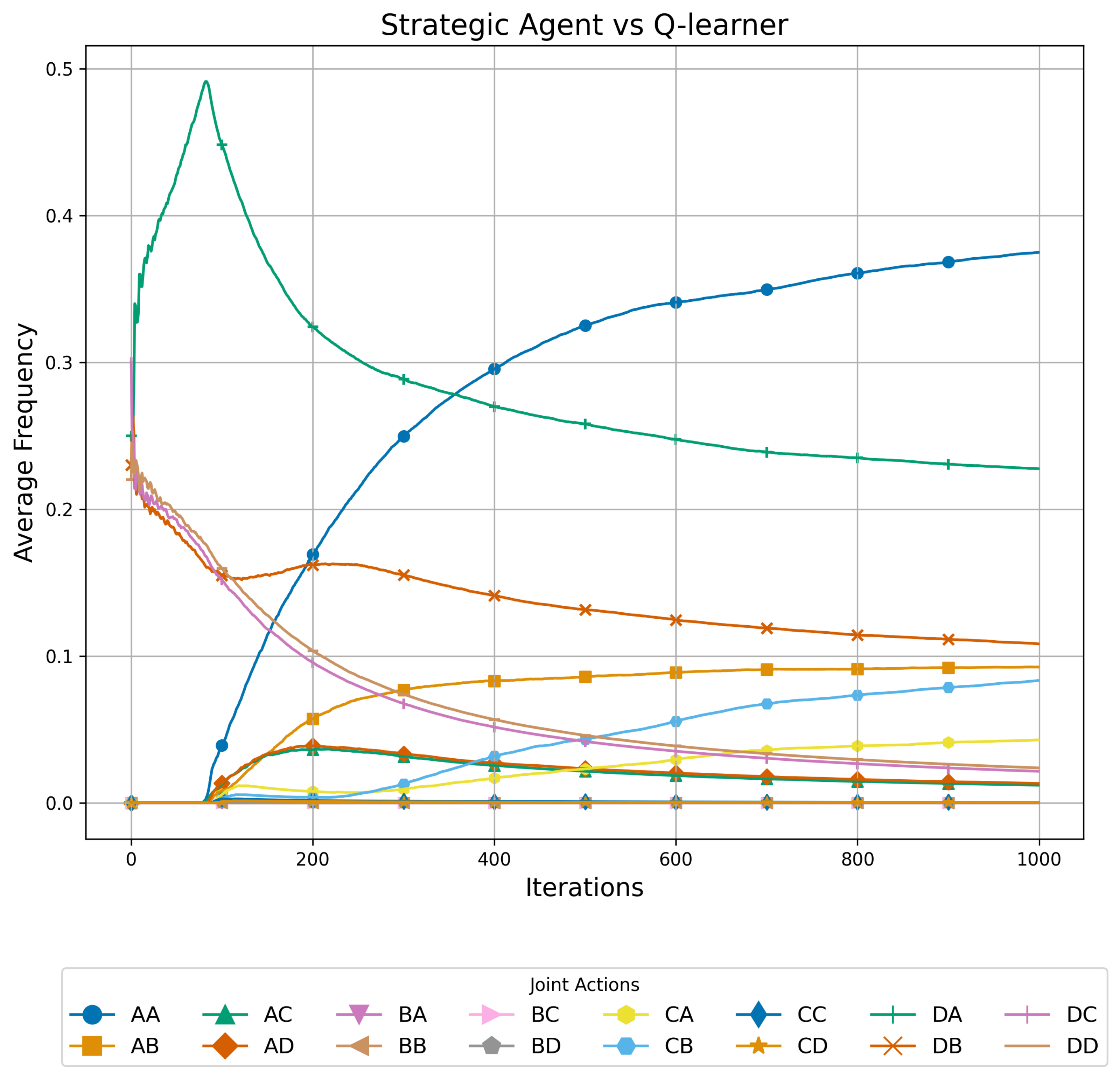}
    \caption{The evolution of the empirical averages of the action profiles for strategic actor vs Q-learner in the two-agent four-action zero-sum game.}
    \label{fig:zero-sum2}
\end{figure}

\section{ILLUSTRATIVE EXAMPLES}\label{sec:example}

We illustrate the performance of \ats~and \nts~interacting with each other in a diverse set of scenarios: $(i)$ a single \at~against a single \nt, $(ii)$ a single \at~against multiple \nts, and $(iii)$ two competing \ats~against a single \nt. We set the parameters of the IQLs deployed by \nts~as $\tau=0.01$ and $\alpha = 0.05$. We set $\gamma = 0.8$ for the utility \eqref{eq:aliceobjective1}. We compute the Markov stationary strategy of \ats~based on the quantization scheme described in Section \ref{sec:approximation}. We quantize the entries of the q-functions to $20$ uniform intervals for the scenarios with two \nts~and $100$ uniform intervals for the rest of the experiments. We run $100$ independent trials and present the mean empirical average of the action profiles as in \cite{ref:Banchio22}. In all these scenarios, \ats~collected \textit{more} payoff (and, therefore, attained higher utility) on \textit{average}, compared to the cases where they also deploy IQLs. 

\textbf{Single \at~against Single \nt.} We focus on the prisoners' dilemma game due to its importance for tacit collusion. Furthermore,  we study 
the performance of the proposed quantization-based method in a randomly generated two-agent zero-sum game where each agent has four actions. In the prisoner's dilemma, under the strategic act of \at, the empirical frequency of cooperation
decreased from $95\%$ to $60\%$ since the \at~pushed the \nt~to play ``C" when she plays ``D", as illustrated in Figs. \ref{fig:pd1} and \ref{fig:pd2}. Therefore, the \at's strategic act exploits \nt's IQL~and increases \at's utility \eqref{eq:aliceobjective1} from $0.1041$ to $0.2051$.\footnote{The utilities are normalized with $(1-\gamma)$ for $10^3$ iterations.} 
Similarly, in the zero-sum game (see Figs. \ref{fig:zero-sum1} and \ref{fig:zero-sum2}), \at~achieved the utility of $0.5340$ when she acted strategically, which is higher than the value $0.4465$ when she deployed IQL. 

When the number of players is increased to three, \at~continues to obtain a higher value compared to the cases where she also deploys IQLs. 

\textbf{Single \at~against Multiple \nts.} We perform experiments
on a single \at~and two \nts~for the cases where the \at~has aligned and misaligned objectives with \nts~and \nts~play a potential game given any action of \at. For the aligned objectives, we observe that when \at~acts strategically by solving the MDP (see Remark \ref{rem:single}), the average value of all agents increases, indicating that \at~drives \nts~to higher payoffs. For the misaligned objectives, the \at's payoff is the opposite of the \nts' potential function. We observe that the \at~can exploit \nts~for its own benefit while \nts~suffer from the \at's strategic act. These results, depicted in Table \ref{tab:potential_games}, quantify how strategically sophisticated agents can manipulate naive agents in potential games both positively and negatively depending on the alignments of their objectives.

\begin{table}[t!]
    \centering
    \renewcommand{\arraystretch}{1.5} 
            \caption{Single \at~vs Multiple \nts}
    \begin{tabular}{c|c|c|c|c}
    & & \textbf{Agent Type} & \textbf{Learning Dynamics} & \textbf{Average Value} \\ \hline\hline
    \multirow{6}{*}{\rotatebox{90}{Aligned Obj.}} & \multirow{3}{*}{\rotatebox{90}{N$\times$N$\times$N}}& N-type & IQL & 0.1593 \\ \cline{3-5}
    && N-type & IQL & 0.1593 \\ \cline{3-5}
    && N-type & IQL & 0.1593 \\ \cline{2-5}
    &\multirow{3}{*}{\rotatebox{90}{N$\times$N$\times$A}}& A-type & DP & 0.2396 \\ \cline{3-5}
    && N-type & IQL & 0.2396 \\ \cline{3-5}
    && N-type & IQL & 0.2396 \\ \hline\hline
    \multirow{6}{*}{\rotatebox{90}{Misaligned Obj.}} & \multirow{3}{*}{\rotatebox{90}{N$\times$N$\times$N}} & N-type & IQL & -0.1470 \\ \cline{3-5}
    && N-type & IQL & 0.1470 \\ \cline{3-5}
    && N-type & IQL & 0.1470 \\ \cline{2-5}
    &\multirow{3}{*}{\rotatebox{90}{N$\times$N$\times$A}} & A-type & DP & -0.0810 \\ \cline{3-5}
    && N-type & IQL & 0.0810 \\ \cline{3-5}
    && N-type & IQL & 0.0810 
    \end{tabular}
    \label{tab:potential_games}
\end{table}
\begin{table}[t!]
    \centering
    \renewcommand{\arraystretch}{1.57}
            \caption{Two Competing \ats~vs Single \nt}
    \begin{tabular}{c|c|c|c|c}
    & \textbf{Agent Type} & \textbf{Payoffs} & \textbf{Learning Dynamics} & \textbf{Average Value} \\ \hline\hline
    \multirow{3}{*}{\rotatebox{90}{N$\times$N$\times$N}} & N-type & $\widetilde{U}$ & IQL & 0.4280 \\ \cline{2-5}
    & N-type & $U$ & IQL & 0.3640 \\ \cline{2-5}
    & N-type & $-U$ & IQL & -0.3640 \\ \hline\hline
    \multirow{3}{*}{\rotatebox{90}{N$\times$A$\times$N}} & N-type & $\widetilde{U}$ & IQL & 0.4875 \\ \cline{2-5}
    & A-type & $U$ & DP & 0.4995 \\ \cline{2-5}
    & N-type & $-U$ & IQL & -0.4995 \\ \hline\hline
    \multirow{3}{*}{\rotatebox{90}{A$\times$A$\times$N}} & N-type & $\widetilde{U}$ & IQL & 0.4854 \\ \cline{2-5}
    & A-type & $U$ & Minimax-DP & 0.5092 \\ \cline{2-5}
    & A-type & $-U$ & Minimax-DP & -0.5092 
    \end{tabular}
    \label{tab:non_potential_games}
\end{table}

\textbf{Two Competing \ats~against Single \nt.}
We let \nt~have an arbitrary payoff matrix, denoted by $\widetilde{U}$, while \ats~have opposite objectives, denoted by $U$ and $-U$. In Table \ref{tab:non_potential_games}, we first report the average values for the case N$\times$N$\times$N, where all agents adopt IQLs. Then, we demonstrate the effect of a strategically sophisticated agent by allowing one to have \at, i.e., act strategically by solving the MDP via DP. Note that \nts~do not play a potential game now given any action of \at, which is different from the example above. Depending on the alignment of their objectives with \at, the strategic act of \at~has positive and negative impacts on them.
 Lastly, we show the impact of competing \ats~deploying minimax-DP to compute their minimax-equilibrium strategy for the two-agent zero-sum SG $\mathcal{Z}$. Interestingly, \nt~benefits from the competition between \ats~when they act strategically. We attribute this to the alignment of \nt's payoff with \at~achieving a higher payoff in A$\times$A$\times$N than N$\times$N$\times$N. Hence, we can observe such intriguing outcomes for competing strategic actors.
 
\section{CONCLUSION}\label{sec:discussion}

We addressed how strategically sophisticated agents can strategize against naive opponents following independent Q-learning algorithms in normal-form games played repeatedly if the strategic actors know the deployed Q-learning algorithms and the underlying game payoffs. We modeled the strategic actors' interactions as an SG as if the Q-learning algorithms are the underlying dynamic environment. However, this resulted in an SG with a continuum state space. We demonstrated the value function specific to Q-learning updates is a Lipschitz continuous function of the continuum state, which is promising for effective approximations via value-based methods. As an example proof of concept, we focused on quantization-based approximation to reduce the problem to a finite SG that can be solved via minimax value iteration when there are two competing strategic actors. We also analytically quantified the approximation level and numerically examined the quantization scheme's performance. 

This work paves the way for understanding the vulnerabilities of learning algorithms against strategic actors from a control-theoretical perspective so that we can design algorithms used in the wild reliably. Possible research directions include \textit{(i)} understanding the vulnerabilities of other widely used algorithms, \textit{(ii)} reducing the capabilities of the strategic actors, and \textit{(iii)} using other approximation methods.

{\appendices
\section{Proof of Lemma \ref{lem:pi}}

Based on the definition \eqref{eq:pizz}, we have
\begin{flalign}
\|\overline{\pi}(z) - \overline{\pi}(\overline{z})&\|_1 = \sum_{b\in B} |\overline{\pi}(z)(b) - \overline{\pi}(\overline{z})(b)| \nn\\
&=\sum_{b\in B}\left|\prod_{i\in\IN} \sigma(q^i)(b^i) - \prod_{i\in\IN}\sigma(\overline{q}^i)(b^i)\right|\nn\\
&\stackrel{(a)}{\leq} \sum_b \left(\prod_{i\neq \ell} \sigma(q^i)(b^i)\right)|\sigma(q^\ell)(b^{\ell}) - \sigma(\overline{q}^{\ell})(b^{\ell})|\nn\\
&\;\;\;\;+\sum_{b}\sigma(\overline{q}^\ell)(b^{\ell})\left|\prod_{i\neq \ell}\sigma(q^i)(b^i) - \prod_{i\neq \ell}\sigma(\overline{q}^i)(b^i)\right|\nn\\
&\stackrel{(b)}{=}\|\sigma(q^\ell)-\sigma(\overline{q}^\ell)\|_1 \nn\\
&\;\;\;\;+ \sum_{b^{-\ell}}\left|\prod_{i\neq \ell} \sigma(q^i)(b^i) - \prod_{i\neq \ell}\sigma(\overline{q}^i)(b^i)\right|,\label{eq:pilong}
\end{flalign}
where $(a)$ follows from the triangle inequality and by adding and subtracting $\left(\prod_{i\neq \ell} \sigma(q^i)(b^i)\right)\sigma(\overline{q}^\ell)(b^\ell)$ for some $\ell\in\IN$, and $(b)$ follows from the fact that the softmax, as described in \eqref{eq:softmax}, is a probability distribution. By following similar lines with \eqref{eq:pilong}, we can show that
\begin{flalign}
\|\overline{\pi}(z) - \overline{\pi}(\overline{z})\|_1 &\leq \sum_{i\in\IN}\|\sigma(q^i)-\sigma(\overline{q}^i)\|_1\nn\\
&\leq \sum_{i}\sqrt{|B^i|} \cdot \|\sigma(q^i)-\sigma(\overline{q}^i)\|_2\nn\\
&\stackrel{(a)}{\leq} \frac{1}{\tau}\sum_{i}\sqrt{|B^i|}\cdot \|q^i-\overline{q}^i\|_2\nn\\
&\leq \frac{\sqrt{|B|}}{\tau} \sum_{i\in\IN}\|q^i-\overline{q}^i\|_1\nn\\
&= \frac{\sqrt{|B|}}{\tau} \cdot \|z-\overline{z}\|_1,
\end{flalign}
where $(a)$ follows from the $\frac{1}{\tau}$-Lipschitz continuity of the softmax function $\sigma(\cdot)$ with respect to $\|\cdot\|_2$, e.g., see \cite[Proposition 4]{ref:Gao17}.

\section{Proof of Lemma \ref{lemma:reward_lip}}
By \eqref{eq:rz}, the left-hand side of \eqref{eq:Lipschitz_reward} can be written as
            \begin{flalign}
                |r^j(z,a)-r^j(\overline{z},a)| &= 
                | u^j(a,\cdot)^T (\piz(z)-\piz(\overline{z})) |\nn\\
                &\stackrel{(a)}{\leq} \|u^j(a,\cdot)\|_\infty \|\piz(z)-\piz(\overline{z})\|_1 \nn\\
                &\stackrel{(b)}{\leq} L_0 \cdot \|z-\overline{z}\|_1, \label{eq:Lipschitz_softmax}
            \end{flalign}
        where $(a)$ follows from the H\"{o}lder inequality, and $(b)$ follows from Lemma \ref{lem:pi} and the definition of $L_0$. 
        
\section{Proof of Lemma \ref{lemma:prob_bound}}

By the definition of the transition kernel $p(\cdot| \cdot)$, as described in Problem \ref{prob:mdp}, we can write the left-hand side of the inequality \eqref{eq:prob_lip} as 
        \begin{flalign}\label{eq:prob_lip_2}
            \Xi := \left| \sum_b \piz(z)(b) \cdot \eta(z_+^{ab}) - \sum_b \piz(\overline{z})(b) \cdot \eta(\overline{z}_+^{ab}) \right|,
        \end{flalign}
        where $z_+^{ab}=\{q_+^{i,ab}\}_{i\in\IN}$ is the next state given the current state and the action profile $(a,b)$.
        Correspondingly, $\overline{z}_+^{ab}=\{\overline{q}_+^{i,ab}\}_{i\in\IN}$ is the next state given the corresponding state $\overline{z}$ and the action profile $(a,b)$.
        
If we add and subtract $\sum_b \piz(z)(b) \cdot \eta(\overline{z}_+^{ab})$, the triangle inequality yields that
            $\Xi \leq \Xi_1 + \Xi_2$,
        where 
        \begin{subequations}
            \begin{flalign}
                &\Xi_1 \coloneqq \left| \sum_b \piz(z)(b) \cdot \left(\eta(z_+^{ab})-\eta(\overline{z}_+^{ab})\right) \right|,\\
                &\Xi_2 \coloneqq \left| \sum_b \Big(\piz(z)(b) - \piz(\overline{z})(b)\Big) \cdot \eta(\overline{z}_+^{ab}) \right|.
            \end{flalign}
        \end{subequations}
Then, we can bound $\Xi_1$ by             
\begin{flalign}
                \Xi_1 &\leq \sum_b \piz(z)(b) \cdot \left| \eta(z_+^{ab})-\eta(\overline{z}_+^{ab})\right|\nn\\
                &\stackrel{(a)}{\leq} K \sum_b \piz(z)(b) \cdot \|z_+^{ab}-\overline{z}_+^{ab}\|_1 \nn\\
                &\stackrel{(b)}{\leq} K\cdot \|z-\overline{z}\|_1,\label{eq:xi1}
            \end{flalign}
            where $(a)$ follows from the $K$-Lipschitz continuity of $\eta(\cdot)$ from the statement of the lemma, and $(b)$ follows from $\|\piz(z)\|_1=1$ and the fact that
            \begin{align}
        &\|z_+^{ab}-\overline{z}_+^{ab}\|_1 = \sum_{i\in\IN}\sum_{b^i\in B^i}|q_+^{i,ab}(b^i) - \overline{q}_+^{i,ab}(b^i)|\nn\\
        &=\sum_{i\in\IN}\left((1-\alpha)|q^i(b^i)-\overline{q}^i(b^i)| + \sum_{\tilde{b}^i\neq b^i} |q^i(\tilde{b}^i)-\overline{q}^i(\tilde{b}^i)|\right)\nn\\
        &\leq \sum_{i\in\IN}\|q^i-\overline{q}^i\|_1.\label{eq:qplusdiff}
        \end{align}
        On the other hand, we can bound $\Xi_2$ by
            \begin{flalign}
                \Xi_2 &\stackrel{(a)}{\leq} \|\piz(z)-\piz(\overline{z})\|_1\cdot  \max_{b\in B}|\eta(\overline{z}_+^{ab})|\nn\\
                &\stackrel{(b)}{\leq} \frac{\sqrt{|B|}}{\tau} M \cdot \|z-\overline{z}\|_1.\label{eq:xi2} 
            \end{flalign}
        where $(a)$ follows from the H\"{o}lder inequality and $(b)$ follows from Lemma \ref{lem:pi} and the boundedness of $\eta(\cdot)$ from the statement of the lemma. Combining \eqref{eq:xi1} and \eqref{eq:xi2} lead to \eqref{eq:prob_lip}.}

\bibliographystyle{IEEEtran}
\bibliography{mybib}

\end{document}